\documentclass[dvipsnames,letterpaper,11pt]{article}
\usepackage[margin=1in]{geometry}
\usepackage{booktabs} 
\usepackage[ruled]{algorithm2e} 
\usepackage[hidelinks]{hyperref}
\usepackage{amsmath,amsthm}
\usepackage{cleveref}
\usepackage{stmaryrd}
\usepackage[longnamesfirst,sort]{natbib}
\usepackage[font={footnotesize}]{caption}
\usepackage{subcaption}
\usepackage{float}
\usepackage{dsfont}
\usepackage{cancel}
\usepackage{xcolor}
\usepackage{soul}
\usepackage{thmtools} 
\usepackage{thm-restate}
\usepackage{comment}

\SetAlFnt{\small}
\SetAlCapFnt{\small}
\SetAlCapNameFnt{\small}
\SetAlCapHSkip{0pt}
\IncMargin{-\parindent}
\usepackage{enumitem}
\usepackage{bbold}
\usepackage{xspace}
\usepackage{setspace}

\newcommand{\NN}{\mathbb{N}}

\newcommand{\EE}{{\bf E}}
\newcommand{\PP}{{\bf P}}

\newcommand{\calF}{\mathcal{F}}

\newcommand{\calJ}{\mathcal{J}}

\newcommand{\ALG}{\textnormal{ALG}}
\newcommand{\OPT}{\textnormal{OPT}}

\newcommand{\PNKref}{\text{\normalfont\hyperref[PNK]{[P$_{m,n,k}$]}}\xspace}
\newcommand{\DNKref}{\text{\normalfont\hyperref[DNK]{[D$_{m,n,k}$]}}\xspace}

\newcommand{\DP}{\textnormal{DP}}

\newtheorem{theorem}{Theorem}
\newtheorem{lemma}{Lemma}

\newtheorem{proposition}{Proposition}

\crefname{claim}{Claim}{Claims}
\crefname{obs}{Observation}{Observation}
\crefname{cons}{constraint}{constraints}
\crefname{eq}{equality}{equalities}
\crefname{eqn}{equation}{equations}
\crefname{ineq}{inequality}{inequalities}

\creflabelformat{cons}{(#2#1#3)}
\creflabelformat{eq}{(#2#1#3)}
\creflabelformat{eqn}{(#2#1#3)}
\creflabelformat{ineq}{(#2#1#3)}

\begin{document}

\title{Competition Versus Complexity in Multiple-Selection\\ Prophet Inequalities
\footnote{This work was partially funded by the ANID (Chile) through Grant FONDECYT 1241846.}
}
\author{Eugenio Cruz-Ossa
\thanks{School of Engineering, Pontificia Universidad Cat\'olica de Chile.}
\and Sebastian Perez-Salazar
\thanks{Department of Computational Applied Mathematics and Operations Research, Rice University, USA.}
\thanks{Ken Kennedy Institute, Rice University, USA.}
\and Victor Verdugo
\thanks{Institute for Mathematical and Computational Engineering, Pontificia Universidad Católica de Chile.} 
\thanks{Department of Industrial and Systems Engineering, Pontificia Universidad Cat\'olica de Chile.}
}
\date{}
\maketitle
\begin{abstract}
Competition complexity formalizes a compelling intuition: rather than refining the mechanism, how much additional competition is sufficient for a simple mechanism to compete with an optimal one? We begin the study of this question in multi-unit pricing for welfare maximization using prophet inequalities. An online decision-maker observes $m \geq k$ nonnegative values drawn independently from a known distribution, may select up to $k$ of them, and aims to maximize the expected sum of selected values. The benchmark is a prophet who observes a sequence of length $n \geq k$ and selects the $k$ largest values. We focus on the widely adopted class of single-threshold algorithms and fully characterize their $(1-\varepsilon)$-competition complexity. Notably, our results reveal a sharp competition-induced phase transition: in the absence of competition, single-threshold algorithms are fundamentally limited to a $1-1/\sqrt{2k\pi}$ fraction of the prophet value, whereas even a $1\%$ multiplicative increase beyond $n$ observations suffices to achieve a $1-\exp(-\Theta(k))$ fraction.
Another notable result happens when $k=1$: we show that the $(1-\varepsilon)$-competition complexity is exactly $\ln(1/\varepsilon)$, fully resolving an open question by Brustle et al. [Math. Oper. Res. 2024]. Our analysis is based on infinite-dimensional linear programming and duality arguments.
\end{abstract}

\thispagestyle{empty}

\section{Introduction}

Competition complexity, which traces back to the work of Bulow and Klemperer~\citep{bulow1996auctions}, quantifies the additional competition required for a suboptimal auction mechanism to achieve the revenue of an optimal mechanism. 
Recently, there has been a growing interest in the (approximate) competition complexity of prophet inequality problems~\citep{ezra2025competition,brustle2024competition,brustle2025competition}, mainly driven by the equivalence between prophet inequalities and posted price mechanisms; see, e.g.,~\cite{chawla2010multi,correa2019pricing,hajiaghayi2007automated}.

In the classic multiple-selection independent and identically distributed (i.i.d.)\ prophet inequality (see, e.g.,~\citep{hill1982comparisons,correa2021posted,brustle2025splitting,jiang2025tightness}), a decision-maker observes, in an online fashion, $n\geq 1$ nonnegative random values $X_1,\ldots,X_n$ drawn from a distribution $F$. Upon observing $X_i$, the decision-maker must irrevocably accept the value or reject it, with no possibility of recalling past rejected values. The decision-maker can select at most $k\geq 1$ values and aims to maximize the expected sum of the values chosen. The decision-maker can solve this problem optimally via dynamic programming, which also provides an optimal online algorithm and attains an expected value of $\DP_{n,k}(F)$. In contrast, a prophet decision-maker (or simply a prophet) observes all values in advance and selects the $k$ largest ones securing an expected value of $\OPT_{n,k}(F)$. Prophet inequalities aim to quantify the \emph{competitive ratio}, defined as the smallest fraction $\DP_{n,k}(F)/\OPT_{n,k}(F)$ over all distributions $F$ and $n$. This worst-case ratio measures the online decision-maker’s lack of foresight in the worst-case scenario. 
For $k=1$, it is known that this value is approximately $0.745$~\citep{correa2021posted} while for $k\geq 2$ the best known analytical lower bound is $1-e^{-k}k^k/k!\approx 1-1/\sqrt{2\pi k}$ (see, e.g.,~\citep{arnosti2023tight}) which is, in fact, attained by a single-threshold algorithm.

The classic (exact) competition complexity question in prophet inequalities instead seeks the \emph{smallest $m$} such that $\DP_{m,k}(F)\geq \OPT_{n,k}(F)$ for any distribution $F$. 
Namely, it asks for the \emph{minimum number of additional observations}, beyond a sequence of length $n$, that the online decision-maker must receive in order to secure an expected selected value at least the one obtained by the prophet. 
For $k=1$, Brustle et al.~\cite{brustle2024competition} show that this exact notion of competition complexity is impossible to satisfy: for any $m\geq n$, there exists a distribution $F$ such that $\DP_{m,1}(F)<\OPT_{n,1}(F)$. Consequently, they instead characterize the (approximate) $(1-\varepsilon)$-competition complexity, which seeks the smallest $m$ such that $\DP_{m,1}(F)\geq (1-\varepsilon)\OPT_{n,1}(F)$ for every $F$. 
They show that it is necessary and sufficient that $m=\phi(\varepsilon)\,n$, where $\phi(\varepsilon) = \int_0^1 (y(1-\ln y) + \varepsilon)^{-1} \mathrm{d}y =\Theta(\ln\ln(1/\varepsilon))$. Although the $(1-\varepsilon)$-competition complexity is well understood for $k=1$, such a characterization for the general multiple-selection setting $k\geq 2$ remains open---even for suboptimal but simpler classes of algorithms.

In the multiple-selection setting, the optimal online algorithm has the following structure: it computes $k \cdot m$ thresholds $\tau_{1,1},\ldots,\tau_{m,k}$ and accepts $X_i$ whenever $X_i \ge \tau_{i,j}$ and, at the time of its observation, there are $j$ remaining selections. A widely adopted subclass of algorithms consists of \emph{single-threshold} algorithms---where $\tau_{1,1}=\cdots=\tau_{m,k}$---due to its managerial interpretability, simplicity, and perceived fairness (see, e.g.,~\citep{arnosti2023tight,perez2025iid,chawla2024static}). 
This raises a natural question: \emph{how much additional competition is required when the decision-maker restricts attention to such simple single-threshold algorithms, rather than the more complex optimal online algorithm, in order to approximate the value of the prophet?} The results in~\cite{brustle2024competition} provide some insights for $k=1$, showing that any single-threshold algorithm requires $m\geq n$ observations such that $m/n=\Omega(\ln(1/\varepsilon))$ but leaving open the exact characterization (see Section~1.5 of~\cite{brustle2024competition}). 
Using the \emph{block model} in~\cite{brustle2025competition} one can also deduce that $c\ln(1/\varepsilon) \leq m/n \leq \log_2(1/\varepsilon)$ for some $c>0$. While this provides useful asymptotic insight into the scaling of $m$ with respect to $n$, it does not yield exact scalings that may be of practical relevance. 

\paragraph{Our contribution.} We provide a complete characterization of the $(1-\varepsilon)$-competition complexity for single-threshold algorithms in the multiple-selection i.i.d.\ prophet inequality {as the value $\beta$ that satisfies $\EE[\min\{ k, \mathrm{Poisson}(\beta\cdot k) \}]=(1-\varepsilon)k$, where $\mathrm{Poisson}(\lambda)$ is a Poisson random variable of parameter $\lambda$}. As a by-product of our analysis, we find that for $k=1$, we must have $\beta=\ln (1/\varepsilon)$. Our results yield several notable insights. First, in the single-selection case, since $\ln (1/(1-0.745))\approx 1.37$, the best single-threshold algorithm achieves the same guarantees as the optimal multi-threshold algorithm while requiring only $37\%$ more observed values. 
This substantially improves on the block-model estimates, which pessimistically suggest that achieving the same guarantees requires 97\% more.
Second, for $k>1$, even a small $1\%$ increase over $n$ ensures that the optimal single-threshold algorithm on a sequence of length $m$ secures a $1-\exp{(-\Theta(k))}$ fraction of the prophet’s value---in stark contrast with the case $m=n$, where the best known bound is $\approx 1-1/\sqrt{2\pi k}$. In line with the original insight of Bulow and Klemperer~\citep{bulow1996auctions}, our results in the multiple-selection setting suggest that expanding the pool of buyers has a larger impact on revenue than designing the optimal posted-price mechanism. {Finally, our results suggest a simple algorithmic approach that is independent of $m$. Given $n$ and $F$, compute a value $x$ such that $F(x)=1-k/n$, and select the first $k$ observed values that are at least $x$. If the algorithm observes a sequence of length $m\geq \beta n$, it achieves an expected value at least a fraction $(1-\varepsilon)$ of the expected offline optimal value for a sequence of length $n$. Moreover, $m=\beta n$ is the smallest sequence length that guarantees this performance for every input distribution.} Our analysis relies on a quantile-based approach and a formulation of the competition complexity problem as an infinite-dimensional linear program, which we use to derive optimal guarantees. We formally introduce the problem in the next subsection and then present our technical results.

\subsection{Problem Formulation}

Let $\mathcal{F}$ be the class of continuous distributions over $[0,\infty)$. Given a distribution $F\in \mathcal{F}$, $n\geq 1$ and $k\in [n]$, we denote by $\OPT_{n,k}(F)=\sum_{i=1}^k\EE[X^{(i)}]$ the value of the prophet which is the sum of the $k$ largest values of $X_1,\ldots,X_n$, independently sampled from $F$. Here $X^{(i)}$ denotes the $i$-th largest ordered statistic. 
We focus on single-threshold algorithms: given $T\geq 0$, accept one-by-one at most $k$ values that are at least $T$. 
A quantile-based single-threshold algorithm is parameterized by a quantile $q\in[0,1]$, with the threshold $T$ chosen such that $F(T)=1-q$; for the family $\mathcal{F}$, the class of single-threshold algorithms coincides with the class of quantile-based single-thresholds.
We note that, for general distributions, our analysis still holds; however, we need to allow randomization by breaking ties at random. One way to achieve this is by adding a random noise to the values, which is a standard procedure in single-threshold problems, e.g.,~\cite{chawla2024static,perez2025iid} and in general prophet inequality problems, e.g.,~\citep{brustle2025splitting,jiang2025tightness}.

For $m\geq k$, let $\ALG_{m,k}(F,q)$ denote the value obtained by a single-threshold algorithm with parameter $q$ on a sequence $X_1,\ldots,X_m$ of length $m$. 
By a simple calculation, the algorithm's reward is given by
$\ALG_{m,k}(F,q) = \EE[\min\{ k,Y\}]\EE[X\mid X\geq F^{-1}(1-q)],$
where $Y$ is a binomial random variable with parameter $m$ and $q$, i.e., $\PP[Y=\ell]=\binom{m}{\ell}q^{\ell}(1-q)^{m-\ell}$ for $\ell\in \{0,\ldots,m\}$.
The single-threshold \emph{$(m,n,k)$-competitive ratio} is the worst-case ratio, over $F\in\mathcal{F}$, between the value of the best single-threshold algorithm and the value of the prophet, i.e.,
\[
\gamma_{m,n,k}=\inf_{F \in \calF} \sup_{q\in [0,1]} \frac{\ALG_{m,k}(F,q)}{\OPT_{n,k}(F)}.
\]
Given $\varepsilon\in [0,1]$, let $\beta_{k,n}(\varepsilon)=\inf_{m\geq k}\left\{m/n :  \gamma_{m,n,k} \geq 1-\varepsilon   \right\}$, i.e., the smallest possible scaling so the $(m,n,k)$-competitive ratio is at least $ 1-\varepsilon$ for a fixed number of values $n$. 
We are interested in characterizing the {\emph{$(1-\varepsilon,k)$-competition complexity} defined as the largest possible scaling for large $n$, i.e., $\beta_k(\varepsilon) = \limsup_{n \to \infty} \beta_{k,n}(\varepsilon)$}. 
We remark that for $\varepsilon=0$ (i.e., the algorithm's reward is required to be as good as the prophet), the impossibility result in~\cite{brustle2024competition} holds for the stronger optimal online algorithm, even in the case $k=1$. 
Since single-threshold algorithms are weaker, this impossibility applies here, i.e., the $(1,1)$-competition complexity is $\beta_1(0)=\infty$.

\subsection{Technical Overview of Our Results}

The first step in our analysis is to characterize exactly the $(m,n,k)$-competitive ratio.
To this end, in Section \ref{sec:LP-formulation} we provide an infinite-dimensional linear programming formulation whose optimal value exactly equals $\gamma_{m,n,k}$.
The LP construction starts from the fact that both $\ALG_{m,k}$ and $\OPT_{n,k}$ can be written linearly in quantile space, after performing the transformation $T=F^{-1}(1-q)$ for the thresholds; the LP characterization is summarized by Proposition \ref{thm:equivalence-PNK-principal}.
Equipped with the linear program, our next step is to solve it to optimality.
In general, infinite-dimensional LPs do not necessarily satisfy strong duality, and therefore, using the existing finite-dimensional machinery directly is not possible.
However, in our case, we establish a dual program in Section~\ref{sec:strong-duality} and prove that it is in strong duality with the LP formulation of $\gamma_{m,n,k}$.
To show this, we first establish that weak-duality holds between the pair (Proposition~\ref{prop:weak-duality}), and then we exhibit explicit primal and dual solutions with equal objective value.
Then, we conclude strong-duality by showing primal and dual feasibility for the solutions (Theorem~\ref{thm:optimal-solution-PNK}).
While dual feasibility is simple to check, primal feasibility is technically involved (Lemma~\ref{prop:supremum-q-d}), and Section \ref{subsec:proof-supremum-q-d} is devoted to this end.

We remark that in Theorem~\ref{thm:optimal-solution-PNK} we are able to get the exact value of $\gamma_{m,n,k}$ by a closed formula that can be easily evaluated.
In particular, using this formula we can directly recover the best possible factors in the absence of competition, i.e., $m=n$: $\gamma_{n,n,1}=1-(1-1/n)^n\to 1-1/e$, and for $k>1$ we get $\gamma_{n,n,k}\approx 1-1/\sqrt{2\pi k}$ by the Stirling's approximation. 
An algorithmic insight of our characterization is that for every $m$, the single-threshold given by the quantile $k/n$ guarantees a competitive ratio of $\gamma_{m,n,k}$ and this is tight.
Using our exact characterization of the $(m,n,k)$-competitive ratio, in Section \ref{sec:CCbounds} {(Theorem~\ref{thm:beta_characterization}) we prove that the $(1-\varepsilon,k)$-competition complexity $\beta_{k}(\varepsilon)$ satisfies the following condition:
\[
\EE[\min\{ k, \mathrm{Poisson}(\beta_k(\varepsilon) \cdot k) \}]= (1-\varepsilon)k.
\]
This characterization allows us immediately to prove that for $k=1$, $\beta_1(\varepsilon)=\ln(1/\varepsilon)$.}
For $k>1$, we provide bounds on $\beta_k(\varepsilon)$ (Theorem~\ref{thm:betak-epsilon-bounds}) derived from concentration arguments. Our bounds imply that attracting any additional constant fraction of buyers suffices to break the $1-1/\sqrt{2\pi k}$ barrier and achieve a $1-\exp(-\Theta(k))$ fraction of the prophet, i.e., we exhibit a transition on the competitive ratio by introducing competition. 

{To showcase the effectiveness of competition, in  Table~\ref{tab:comparison_competitive_ratios}, we show how many extra values we need over $n=1000$ so we can actually match the best known upper bounds on the competitive ratio for the optimal multi-threshold algorithm in~\cite{brustle2025splitting} for small values of $k$. 
As we can see, it suffices with $\approx 37.6\%$ when $k=1$, and we just need $\approx 24.4\%$ when $k=5$.
The values in the second row of Table~\ref{tab:comparison_competitive_ratios} are obtained by evaluating $\beta_{n,k}(\varepsilon)$ with $n=1000$ and $\varepsilon$ equal to the upper bound gap for the corresponding $k$, e.g., for $k=1$ we have $\varepsilon=1-0.7474$.
We note that our bounds are obtained using the arguably simplest possible policy, whereas the optimal online algorithm requires multi-threshold strategies, in addition to the intrinsic difficulties of computing these thresholds. }
\begin{table}[htbp]
  \centering
  \begin{tabular}{l|ccccc}
    \hline
     & $k=1$ & $k=2$ & $k=3$ & $k=4$ & $k=5$ \\
    \hline
    Upper bound~\citep{brustle2025splitting} & 0.7474 & 0.8372 & 0.8742 & 0.8949 & 0.9086 \\
    Comp. complexity & 1.376 & 1.330 & 1.293 & 1.265 & 1.244 \\
    \hline
  \end{tabular}
  \caption{Upper bounds on the optimal multi-threshold algorithm (first row) and $\beta_{1000,k}$ (second row).}
  \label{tab:comparison_competitive_ratios}
\end{table}

{Although competition complexity aims to quantify the benefit of access to additional observations, typically assuming $m\geq n$, our results remain valid for all $m,n\geq k$. Thus, for any $\varepsilon\in(0,1)$, $\beta_k(\varepsilon)$ quantifies the number of observations needed relative to $n$ to guarantee a competitive ratio of at least $1-\varepsilon$: it represents a \emph{blow-up factor} when $\beta_k(\varepsilon)\geq1$ and a \emph{fraction} of $n$ when $\beta_k(\varepsilon)<1$. For example, when $k=1$ and $\varepsilon=1/2$, we obtain
$\beta_1(1/2)=\ln(2)\approx0.693$.}

\subsection{Related Work}

Prophet inequalities were introduced by Krengel and Sucheston~\citep{krengel1977semiamarts}. In the general single-selection case, the competitive ratio is $1/2$, and it is attained by a single-threshold algorithm~\citep{samuel1984comparison,wittmann1996superprophet,kleinberg2019matroid}. 
In recent decades, prophet inequalities have attracted considerable attention due to their applicability in mechanism design, as approximation guarantees for posted-price mechanisms are equivalent to competitive ratios in prophet inequality problems in several settings~\citep{chawla2010multi,correa2019pricing}; see~\citep{correa2019recent} for a recent survey. For the i.i.d.\ prophet inequality problem introduced by Hill and Kertz~\citep{hill1982comparisons}, while the competitive ratio for $k=1$ case is fully characterized by the works~\citep{correa2021posted,hill1982comparisons,kertz1986stop}, the exact competitive ratio for multiple-selection remains open. 
Recent works have provided sharp bounds on the competitive ratio~\cite{jiang2025tightness,brustle2025splitting}, although the best known analytical lower bound is $1-1/\sqrt{2\pi k}$ via a single-threshold algorithm; e.g.,~\citep{arnosti2023tight,chakraborty2010approximation,beyhaghi2021improved}.

Our work examines the power of fixed-price mechanisms versus optimal mechanisms through the lens of prophet inequalities, placing it within the broader literature on simple versus optimal mechanisms~\citep{hartline2009simple}. Similar ideas have been explored extensively in auction theory~\citep{feldman201899,eden2016competition,beyhaghi2019optimal,bulow1996auctions}. Most closely related to our work are recent studies of competition complexity in single-selection prophet inequality settings~\citep{brustle2024competition,brustle2025competition,ezra2025competition,ezra2024competition}; we extend this line of research to the multiple-selection setting. 

The quantile approach to i.i.d.\ prophet inequalities (see, e.g.,~\citep{perez2025iid,correa2021posted,brustle2025splitting}) allows the competitive ratio to be reformulated as a linear program~\citep{perez2025iid,brustle2025splitting}. {The work~\cite{perez2025iid} shows how linear programs can capture restricted classes of algorithms (single-threshold, few thresholds, dynamic solutions) for the i.i.d.\ single-selection problem, whereas~\cite{brustle2025splitting} provide an exact characterization of the competitive ratio achieved by optimal algorithms for the multiple-selection prophet inequality. Our work draws on ideas from both papers: we study restricted classes of algorithms for multiple selection and extend the approach to instances with different sequence lengths to characterize the $(m,n,k)$-competitive ratio that we later use for the characterization of the competition complexity. Linear programs over the class of all algorithms are generally difficult to solve, and prior work often derives bounds rather than exact optimal solutions (see, e.g.,~\cite{brustle2025splitting}). In contrast, we obtain exact solutions to our formulations, in part because we restrict attention to single-threshold algorithms. Other linear programming formulations have also been developed for prophet inequalities; see, e.g.,~\cite{jiang2025tightness}.}

The work of~\cite{brustle2024competition}, which initiated the study of competition complexity in prophet inequality problems, uses a convex program in the value space to characterize competitive ratios, which it then uses to analyze competition complexity. {Their formulation relies on a sequence of nontrivial reductions (see also~\cite{perez2024optimal} for a similar approach in a related model). More specifically, they construct an optimization problem parameterized by the dynamic programming values, characterize the space of worst-case distributions, and finally construct an equivalent finite-dimensional convex optimization problem. However, it is unclear how to adapt these reductions to single-threshold algorithms. By following their approach, we cannot obtain a characterization via first-order optimality conditions, as the resulting optimization problem is no longer convex; in particular, first-order optimality conditions are no longer sufficient. With our linear-programming framework, we bypass these difficulties, and optimality follows directly from duality arguments. We remark that the approaches used by other works in the literature (e.g., \cite{brustle2025competition,ezra2024competition,ezra2025competition}) do not give lossless approximation guarantees as we do in this work; our linear-programming formulation characterizes exactly the worst-case approximation factor, which ultimately yields the exact competition complexity guarantees.}

\section{A Linear Programming Formulation}\label{sec:LP-formulation}

{In this section, we introduce the main technical tool used to show that the $(m,n,k)$-competitive ratio is
$\EE\left[\min\left\{k,\mathrm{Binom}\left(m,k/n\right)\right\}\right]/k$
(Theorem~\ref{thm:optimal-solution-PNK} in Section~\ref{sec:strong-duality}), where $\mathrm{Binom}(m,q)$ is a binomial random variable with parameter $m$ and $q$.
This characterization will enable our study of competition complexity in Section~\ref{sec:CCbounds}.
}

Our formulations are quantile-based; thus, we reformulate all relevant quantities in terms of quantiles. Given $k,m\in\NN$ with $m\geq k$, define the function $Q_{m,k}:[0,1]\to\mathbb{R}$ by
\[
Q_{m,k}(q)
=
k-\sum_{\ell=0}^{k-1}(k-\ell)
\binom{m}{\ell}q^\ell(1-q)^{m-\ell}.
\]
In the following proposition, we provide useful expressions for the expected values attained by single-threshold algorithms and the prophet, together with several properties of $Q_{m,k}$.
\begin{proposition}\label{prop:expected-value-algorithm}
    Given $F \in \calF$, $q \in (0,1)$, $k \in \NN$, and $n,m \in \NN$ such that $n,m \geq k$, the following holds with $r(u) = F^{-1}(1-u)$:
    \begin{enumerate}[label=\normalfont(\alph*)]
    \item $\ALG_{m,k}(F,q) = ({Q_{m,k}(q)}/{q}) \int_0^q r(u) du$.\label{characterization-sol-alg}
    \item $\OPT_{n,k}(F) = \int_0^1 g_{n,k}(u) r(u) \, \mathrm{d}u,$
    where $g_{n,k}(u) = \sum_{\ell = n - k + 1}^n \ell \binom{n}{\ell} (1-u)^{\ell-1} u^{n-\ell}$.\label{characterization-sol-opt}
          \item When $Y \sim \normalfont\text{Binomial}(m, q)$, we have
      $Q_{m,k}(q) = k - \EE[(k-Y)_+] = \EE[\min\{k, Y\}].$\label{characterization-sol-a}
      \item $Q_{m,k}'(q) = m \sum_{i=0}^{k-1} \binom{m-1}{i} q^i (1-q)^{m-i-1} \geq 0$ for $q \in [0,1]$, i.e., $Q_{m,k}$ is increasing.\label{characterization-sol-b}
      \item $Q_{m,k}''(q) = -m(m-1)\binom{m-2}{k-1}q^{k-1}(1-q)^{m-k-1} \leq 0$ for $q \in [0,1]$, i.e., $Q_{m,k}$ is concave.\label{characterization-sol-c}
    \end{enumerate}
\end{proposition}
\begin{proof}
We have $\ALG_{m,k}(F,q)=\EE[\min\{k,Y\}]\EE[X\mid X > F^{-1}(1-q)]$, where $Y\sim \normalfont\text{Binomial}(m, q)$.
For the second term in the product, we have $ \EE[X\mid X > F^{-1}(1-q)]=({1}/{q}) \int_0^q r(u) du$ by direct integration and changing variables.
    For the first term, we have
    \begin{align*}
        \EE[\min\{k,Y\}] &= k\cdot \PP[ Y \geq k] + \sum_{\ell = 0}^{k-1} \ell\cdot \PP[Y = \ell]= k - \sum_{\ell = 0}^{k-1} (k - \ell) \binom{m}{\ell} q^\ell (1-q)^{m-\ell}= Q_{m,k}(q),
    \end{align*}
    which concludes part \ref{characterization-sol-alg}.
    For point \ref{characterization-sol-opt}, a proof can be found in, e.g., \cite[Proposition 1]{brustle2025splitting}.
For part \ref{characterization-sol-a}, the first equality holds by identifying the expected value of the positive part, and the second equality is by using the equality $\min\{k, Y\} = k - (k - Y)_+$. 
To show \ref{characterization-sol-b} and \ref{characterization-sol-c}, first note that $k \leq m$ and $\min\{k, Y\} = \sum_{i = 1}^{k} \mathds{1}_{\{Y \geq i\}}$, therefore,
$Q_{m,k}(q) = \mathbb{E}\left[\min\{k, Y\}\right] = \sum_{i = 1}^{k} \PP[Y \geq i]=g_i(q)$, where
$g_i(q) =\sum_{\ell = i}^{m} \binom{m}{\ell} q^{\ell} (1-q)^{m - \ell}$.
From here, by taking derivatives and rearranging the telescopic summations, we get the following expressions for $Q_{m,k}'$ and $Q_{m,k}''$:
\begin{align*}
Q_{m,k}'(q) = \sum_{i=1}^{k} g_i'(q)&=\sum_{i=1}^{k}\sum_{\ell = i}^{m} \left[ \binom{m}{\ell} \ell q^{\ell-1} (1-q)^{m - \ell} - \binom{m}{\ell} (m - \ell) q^{\ell} (1-q)^{m - \ell - 1} \right]\\
&= m\sum_{i=1}^{k}\sum_{\ell = i}^{m} \left[ \binom{m - 1}{\ell - 1} q^{\ell-1} (1-q)^{m - \ell} - \binom{m - 1}{\ell} q^{\ell} (1-q)^{m - \ell - 1} \right]\\
  &= m \sum_{i=0}^{k-1} \binom{m - 1}{i} q^i (1 - q)^{m - (i+1)},\\
Q_{m,k}''(q)&= m \sum_{i=0}^{k-1}\binom{m - 1}{i} \left( iq^{i-1}(1-q)^{m - (i+1)} - (m - 1 - i)q^{i}(1-q)^{m - (i+2)} \right)\\
&= m (m - 1) \left[ \sum_{i=1}^{k-1}\binom{m - 2}{i - 1}q^{i-1}(1-q)^{m - i - 1} - \sum_{i=0}^{k-1}\binom{m - 2}{i}q^{i}(1-q)^{m - i - 2} \right]\\
&=-m (m - 1) \binom{m - 2}{k-1}q^{k-1}(1-q)^{m - k - 1},
\end{align*}
concluding \ref{characterization-sol-b} and \ref{characterization-sol-c}.
\end{proof}

\paragraph{The LP formulation.} {Having expressed both $\ALG_{m,k}(F,q)$ and $\OPT_{n,k}(F)$ in terms of the quantile function $r(q)=F^{-1}(1-q)$, we formulate the $(m,n,k)$-competitive ratio as the following infinite-dimensional linear program over the distributional closure of the cone of nonnegative, non-increasing functions. This closure includes generalized functions such as Dirac delta distributions. To avoid notational clutter, we refer to these generalized functions simply as functions and retain the standard integral notation.}
\begin{align}
\inf \quad
& d
\tag*{\normalfont{\mbox{[P$_{m,n,k}$]}}}
\label{PNK}\\
\text{s.t.}\quad
& d \geq
\frac{Q_{m,k}(q)}{q}
\int_{0}^{q} r(u)\,du
\quad \text{for all }q\in(0,1],
\label{PNK1}\\
& \int_{0}^{1} r(u)g_{n,k}(u)\,du=1,
\label{PNK2}\\
& f\geq0\text{ and non-increasing.}
\label{PNK4}
\end{align}
Intuitively, {Constraint~\eqref{PNK2} normalizes the prophet value so that $\OPT_{n,k}(F)=1$. When $F$ has finite expectation, this normalization can be imposed without loss of generality because rescaling the values changes $\ALG_{m,k}(F)$ and $\OPT_{n,k}(F)$ by the same factor and therefore leaves their ratio unchanged. If $F$ has infinite expectation, the ratio equals one so it is enough to focus on $\OPT_{n,k}(F)=1$.} Constraints~\eqref{PNK1} provide a linearized formulation of $\sup_{q\in [0,1]} \ALG_{m,k}(F,q)$ through the auxiliary variable $d$. We summarize this characterization in the following proposition.
\begin{proposition}\label{thm:equivalence-PNK-principal}
  For every $k \in \NN$ and $n,m \in \NN$ such that $n,m \geq k$, the optimal value of \PNKref is the $(m, n, k)$-competitive ratio $\gamma_{m,n,k}$.
\end{proposition}
\begin{proof}
We denote by $v$ the value of the $(m,n,k)$-competitive ratio.
Let $F \in \calF$ be any distribution and consider $h(u)=F^{-1}(1-u)$.
Define $(r,d_F)$ as follows: $r(u) = {h(u)}/{\int_0^1 g_{n,k}(z) h(z) dz}$ and let
$d_F = \sup_{q \in \left[0,1\right]} ({Q_{m,k}(q)}/{q})\int_{0}^q r(u) du.$
We have that $(r, d_F)$ is feasible for \ref{PNK}. 
Indeed, note that $r$ is non-increasing and non-negative, since $h$ is non-increasing and non-negative as the inverse of a distribution with positive support, i.e., \eqref{PNK4} are satisfied.
Furthermore,
\[
\int_0^1 r(u) g_{n,k}(u) du = \int_0^1 \frac{h(u)}{\int_0^1 g_{n,k}(z) h(z) dz} \cdot g_{n,k}(u) du = 1,
\]
therefore \eqref{PNK2} holds, and \eqref{PNK1} is satisfied by construction.
On the other hand, we have
\begin{align*}
d_F &= \sup_{q \in \left[0,1\right]}\frac{Q_{m,k}(q)}{q}\cdot  \frac{ \int_{0}^q r(u) du}{\int_0^1 g_{n,k}(u) r(u) du}\\
&=\sup_{q \in \left[0,1\right]}\frac{Q_{m,k}(q)}{q}\cdot  \frac{ \int_{0}^q h(u) du}{\int_0^1 g_{n,k}(u) h(u) du} = \sup_{q \in \left[0,1\right]}\frac{\ALG_{m,k}(F,q)}{\OPT_{n,k}(F)},
\end{align*}
where the last equality holds by Proposition \ref{prop:expected-value-algorithm}\ref{characterization-sol-alg}-\ref{characterization-sol-opt}.
By taking the infimum over $F\in \calF$, we conclude that the optimal value of \ref{PNK} is at most $v$.

We will show next that the converse inequality also holds.
Consider a feasible solution $(r,d)$. 
Let $Q$ be a random value uniformly distributed in $[0,1]$ and let $F_r$ be the cumulative distribution function of $r(Q)$. 
Since $f$ is a non-increasing function, by standard properties of quantile functions (see, e.g.,~\cite[Chapter~II]{devroye1986non}) we have
$r(x) = F_r^{-1}(1-x)$ for every $x \in [0,1]$.
Then, by Proposition \ref{prop:expected-value-algorithm},
\[
\OPT_{n,k}(F_r) = \int_0^1 g_{n,k}(u) F_r^{-1}(1-u) du = \int_0^1 g_{n,k}(u) r(u) du=1,
\]
\[
\ALG_{m,k}(F_r,q) = \frac{Q_{m,k}(q)}{q} \int_{0}^q F_r^{-1}(1-u) du = \frac{Q_{m,k}(q)}{q} \int_{0}^q r(u) du,
\]
where in the computation of $\OPT_{n,k}$ we also used that $(r,d)$ satisfies constraint \eqref{PNK2}.
Therefore,
\[
\sup_{q\in [0,1]}\frac{\ALG_{m,k}(F_r,q)}{\OPT_{n,k}(F_r)} = \sup_{q\in [0,1]}\frac{Q_{m,k}(q)}{q} \int_{0}^q r(u) du\le d,
\]
where the inequality holds since $(r,d)$ is feasible for \ref{PNK}.
By taking the infimum over feasible solutions $(r,d)$, we have 
\begin{align*}
v&\le \inf\left\{\sup_{q\in [0,1]}\frac{\ALG_{m,k}(F_r,q)}{\OPT_{n,k}(F_r)}:(r,d)\text{ feasible for }\ref{PNK}\right\}\\
&\le \inf\Big\{d:(r,d)\text{ feasible for }\ref{PNK}\Big\}=\text{optimal value of }\ref{PNK},
\end{align*}
{where the first inequality holds since $r$ is integrable, and then it can be approximated by a smooth and strictly decreasing function $\hat{r}$ (see, e.g., Chapter 3 in~\cite{rudin1987real}); from here, it is possible to find a continuous distribution $\hat F\in \calF$ such that $\hat{r}(q)=\hat F^{-1}(1-q)$ and $\sup_{q\in [0,1]} \ALG_{m,k}(\hat F,q) \leq d +\varepsilon$ with $\varepsilon>0$ arbitrarily small.
We conclude that $v$ is equal to the optimal value of \ref{PNK}.}
\end{proof}

\section{Solving the Linear Program via Duality}\label{sec:strong-duality}
We remark that establishing the $(m,n,k)$-competitive ratio is, by Proposition \ref{thm:equivalence-PNK-principal}, equivalent to computing the optimal value of \ref{PNK}.
In what follows, we derive a closed-form for the optimal value of \ref{PNK} by explicitly constructing an optimal solution via using strong duality. 
We start by formulating a dual problem, and we establish that weak duality holds between the pair of programs (Proposition \ref{prop:weak-duality}).
We then exhibit primal-dual feasible solutions with equal objective values, thereby establishing strong duality. 
The main technical step is establishing primal feasibility~(Lemma~\ref{prop:supremum-q-d}).
We remark that strong duality does not hold in general for infinite-dimensional LP, which, in particular, prevents the use of the finite-dimensional duality machinery in a direct way.
Our result is summarized in the following theorem.
\begin{theorem}\label{thm:optimal-solution-PNK}
  For every $k \in [n]$ and $n,m \in \NN$ such that $n,m \geq k$, the $(m,n,k)$-competitive ratio is 
  \begin{equation*}\gamma_{m,n,k}=\frac{Q_{m,k}(k/n)}{k}=\frac{1}{k}\left(k - \sum_{\ell = 0}^{k-1} (k-\ell)\binom{m}{\ell} \left( \frac{k}{n} \right)^{\ell} \left( 1 - \frac{k}{n} \right)^{m - \ell}\right).
  \end{equation*}
\end{theorem}
{For a fixed $n$, since the function $Q_{m,k}(k/n)$ is increasing in $m$, we can find the smallest $m\geq k$ such that $\gamma_{m,n,k}\geq 1-\varepsilon$. The theorem guarantees that for such $m$ the algorithm that uses $q=k/n$ as a quantile, i.e., for a distribution $F$, the single-threshold algorithm is $F^{-1}(1-k/n)$, always guarantees a fraction $1-\varepsilon$ of $\OPT_{n,k}(F)$. 
This quantile is independent of $m$, and when $k=1$ and $m=n$, it matches the well-known result that the best single-threshold is attained by $q=1/n$.} 

The rest of this section proves Theorem~\ref{thm:optimal-solution-PNK}.
We start by introducing a dual problem:
\begin{align}
\sup \quad &v \tag*{\normalfont{\mbox{[D$_{m,n,k}$]}}}\label{DNK}\\
\text{s.t.} \quad &\int_{0}^{1} \alpha(q)\,dq = 1, \label{DNK1} \\
& \quad v g_{n,k}(u) + \frac{d\eta(u)}{du} \le \int_{u}^{1} \alpha(q) \frac{Q_{m,k}(q)}{q} dq \quad \text{for all } u \in [0,1] \text{ a.e.}, \label{DNK2} \\
& \quad \alpha: [0,1] \to \mathbb{R}_+, \; \eta: [0,1] \to \mathbb{R}_+\;\text{absolutely continuous},\; \eta(0) =\eta(1)= 0.\label{DNK3}
\end{align}
The following proposition proves that weak duality holds between \PNKref and \DNKref.
\begin{proposition}\label{prop:weak-duality}
  The optimal value of \PNKref is at least the optimal value of \DNKref.
\end{proposition}
\begin{proof}
  Let $(r, d)$ be a feasible solution to the problem \PNKref and $(\alpha, \eta, v)$ be a feasible solution to the problem \DNKref. 
  Then, integration by parts implies that
  \begin{align}
    \int_0^1 r(u) \frac{d\eta}{du}(u) du &= - \int_{0}^1 \eta(u) \frac{dr}{du}(u) du \geq 0,
    \label{eq:integration-by-parts}
  \end{align}
  where the equality follows from the border conditions in \eqref{DNK3} and the inequality holds since $\eta \geq 0$ and $r$ is non-increasing, by \eqref{DNK3} and \eqref{PNK4}, respectively. Then, multiplying by $r(u)$ and integrating over $[0,1]$ we get
  \begin{align*}
  v &\leq v + \int_0^1 r(u) \frac{d\eta}{du}(u) du\\
  &= v \int_0^1 g_{n,k}(u) r(u) du + \int_0^1 r(u) \frac{d\eta}{du}(u) du \\
  &= \int_0^1 r(u) \left( v g_{n,k}(u) + \frac{d\eta}{du}(u) \right) du\\
  &\leq \int_0^1 r(u) \int_{u}^1 \alpha(q) \frac{Q_{m,k}(q)}{q} dq\ du= \int_0^1 \alpha(q) \frac{Q_{m,k}(q)}{q} \int_{u}^1 r(u) du\ dq\leq \int_0^1 \alpha(q) dq \cdot d = d,
  \end{align*}
  where the first inequality follows from \eqref{eq:integration-by-parts}, the first equality follows from \eqref{PNK2}, the second inequality from \eqref{DNK2}, the last inequality from \eqref{PNK1}, and the last equality from \eqref{DNK1}. 
  We conclude that $v \leq d$.
\end{proof}
Now we show that strong duality holds between \PNKref and \DNKref. 
To do so, we construct explicit solutions to the primal-dual pair that share the same objective value; by Proposition \ref{prop:weak-duality}, this implies strong duality after establishing feasibility for each corresponding program. 
More specifically, for the primal \ref{PNK}, consider the pair $(r,d)$ given by 
\begin{align*}
    u\mapsto r(u) &= \delta(u) A(n,m, k) + \mathds{1}_{(0,1]}(u) B(n,m, k),\text{ and }d = \frac{Q_{m,k}(k/n)}{k},
  \end{align*}
  with $A(n,m,k) =  k Q_{m,k}'(k/n)/{(n^{2} Q_{m,k}(k/n))}$ and
    $B(n,m,k) = {1}/{k}-Q_{m,k}'(k/n)/{(n Q_{m,k}(k/n))}.$
    For the dual \ref{DNK}, consider $(\alpha,v,\eta)$ given by
    \begin{align*}
    u\mapsto \alpha(u) &= \delta_{{k}/{n}}(u),\; v = \frac{Q_{m,k}(k/n)}{k} ,\text{ and }\\
    \eta(u) &= \begin{cases}
      \frac{n}{k} Q_{m,k}(k/n)u - v\int_0^u g_{n,k}(s) ds & \text{if } u \in [0, {k}/{n}],\\
      Q_{m,k}(k/n) - v\int_0^u g_{n,k}(s) ds & \text{if } u \in ({k}/{n}, 1],
    \end{cases}
  \end{align*}
  where $\delta_{k/n}$ is a Dirac delta at $k/n$.
  Observe that, by construction, both solutions $(r,d)$ and $(\alpha,v,\eta)$ have the same objective value in the primal and dual, respectively.
  The most challenging part is showing primal feasibility for $(r,d)$; this is summarized in the following technical lemma whose full proof is deferred to Section \ref{subsec:proof-supremum-q-d}.
\begin{lemma}\label{prop:supremum-q-d}
  For every $n \in \NN$, $k \in [n]$ and $m \in \NN$ such that $m \geq n$, we have
  \[
    \frac{Q_{m,k}(k/n)}{k} = \sup_{q \in [0,1]} \left\{ \frac{Q_{m,k}(q)}{q} \int_{0}^q r(u) \, du \right\},
  \]
  where $r(u) = \delta(u) A(n, m, k) + \mathds{1}_{(0,1]}(u) B(n, m, k)$.
\end{lemma}
\begin{proof}[Proof of Theorem \ref{thm:optimal-solution-PNK}]
  Let us first prove that $(r, d)$ is a feasible solution for \ref{PNK}.
  By Lemma \ref{prop:supremum-q-d}, we have that $(r,d)$ satisfies constraint \eqref{PNK1}.
  We now verify the non-negativity of $r$. 
  To do so, we have to ensure that both $A(n,m,k)$ and $B(n,m,k)$ are non-negative. 
  The numerator of $A(n, m, k)$ is non-negative since $Q'_{m,k}\ge 0$ by Proposition \ref{prop:expected-value-algorithm}\ref{characterization-sol-b} and for the denominator we have $n^2 Q_{m,k}(k/n) = n^2\left(k - \EE[(k - Y)_+]\right)\geq 0,$
  where the first equality holds from Proposition \ref{prop:expected-value-algorithm}\ref{characterization-sol-a} with $Y\sim \normalfont\text{Binomial}(m, k/n)$, and the inequality is just consequence of $(k-Y)_+\le k$. 
  Therefore, $A(n, m, k) \geq 0$. 
  To prove that $B(n, m, k) \geq 0$, we use the fact that $Q_{m,k}$ is concave, proved in the Proposition \ref{prop:expected-value-algorithm}\ref{characterization-sol-c}; since $Q_{m,k}(0) = 0$ the concavity of $Q_{m,k}$ implies $q Q'_{m,k}(q) \leq Q_{m,k}(q)$ for all $q \in [0,1]$, and therefore, taking $q = k/n$ we have
  \begin{align*}
  B(n, m, k) &= \frac{1}{k} - \frac{Q_{m,k}'(k/n)}{n Q_{m,k}(k/n)}\geq \frac{1}{k} - \frac{1}{n (k / n)} = 0.
  \end{align*}
  Then, $r$ is non-negative, and by construction $r$ is non-increasing, i.e., \eqref{PNK4} holds. 
  
  Now we verify that $(r,d)$ satisfies \eqref{PNK2}.
  Observe that by construction, we have
  \begin{align*}
  \int_0^1 r(u) g_{n,k}(u) du &= \int_0^1 \delta(u) A(n, m, k) g_{n,k}(u) du + \int_0^1  B(n, m, k) g_{n,k}(u) du\\
  &= A(n, m, k) g_{n,k}(0) + B(n, m, k) \int_0^1 g_{n,k}(u) du\\
  &=A(n,m,k)\cdot n+B(n,m,k)\cdot k=1,
  \end{align*}
  where the third equality holds by simply evaluating $g_{n,k}(0)$ and the fact that $\int_0^1g_{n,k}(u)du=k$, and the last equality follows by directly replacing with the values of $A(n,m,k)$ and $B(n,m,k)$.
  Therefore, constraint \eqref{PNK2} holds, which concludes the feasibility proof for $(r,d)$.
  
  Now verify that $(\alpha,v,\eta)$ is feasible for \ref{DNK}.
  It is clear that both $\alpha$ and $v$ are non-negative.
  By simple evaluation, we also have that $\eta(0)=\eta(1)=0$, but we need to show that $\eta$ is non-negative. 
  The derivative of $\eta$ is $({n}/{k}) Q_{m,k}({k}/{n}) - vg_{n,k}(u)$ if  $u \in [0, {k}/{n}]$ and $- vg_{n,k}(u)$ if $u \in ({k}/{n}, 1]$.
  In particular, since $\eta(1)=0$, we have that $\eta$ is non-negative on the interval $({k}/{n}, 1]$.
  On the other hand, recall that $\eta(0)=0$ and to prove that $\eta$ is increasing over $u \in [0, {k}/{n}]$ we use the fact that for all $u \in [0,1]$ it holds that $g_{n,k}(u) \leq n\big( (1-u) + u \big)^{n-1}=n$.
  Now, we have
  \begin{align*}
    \frac{d\eta(u)}{du} &= \frac{n}{k}Q_{m,k}({k}/{n}) - vg_{n,k}(u)\\
    &= Q_{m,k}(k/n) \left( \frac{n}{k} - \frac{v}{Q_{m,k}(k/n)}g_{n,k}(u) \right)\\
    &\geq Q_{m,k}(k/n) \left( \frac{n}{k} - \frac{v}{Q_{m,k}(k/n)}n \right)= Q_{m,k}(k/n) \left( \frac{n}{k} - \frac{n}{k}  \right) = 0,
  \end{align*}
  where the last equality follows since $v = Q_{m,k}({k}/{n})/k$. 
  We conclude that $\eta$ is non-negative, that is, \eqref{DNK3} holds.
  Note that \eqref{DNK1} is also directly satisfied since $\alpha$ is a Dirac delta. 
  Finally, to study \eqref{DNK2}, observe that
  \begin{align*}
    v g_{n,k}(u) + \frac{d\eta(u)}{du} &= v g_{n,k}(u) + 
      \Big(\frac{n}{k} Q_{m,k}\Big(\frac{k}{n}\Big) - vg_{n,k}(u)\Big)\mathds{1}_{[0, {k}/{n}]}(u)- vg_{n,k}(u)\mathds{1}_{({k}/{n}, 1]}(u)\\
    &= 
      \frac{n}{k} Q_{m,k}\Big(\frac{k}{n}\Big)\mathds{1}_{[0, {k}/{n}]}(u)= \int_u^1 \alpha(q) \frac{Q_{m,k}(q)}{q}  dq,
  \end{align*}
  therefore, \eqref{DNK2} is satisfied. We conclude that $(\alpha, v, \eta)$ is a feasible solution to the problem \DNKref.
  In summary, we have a primal-dual feasible pair with equal objective values, and as weak-duality holds by Proposition \ref{prop:weak-duality}, we conclude that strong duality holds with optimal value ${Q_{m,k}(k/n)}/{k}$.
\end{proof}

\subsection{Proof of Lemma~\ref{prop:supremum-q-d}}
\label{subsec:proof-supremum-q-d}

Consider the function $\phi:[0,1] \to \mathbb{R}_+$ as $\phi(q) = ({Q_{m,k}(q)}/{q})\int_0^q r(u)\,du$ for $q\in (0,1]$, and in $q=0$ we define $\phi(0) = \lim_{q \downarrow 0} \phi(q) = m A(n,m,k)$. 
To prove Lemma~\ref{prop:supremum-q-d}, we will show that the function $\phi$ is quasiconcave.
For notation simplicity, in what follows we omit the arguments and subscripts as they remain fixed during this section, i.e., we let $A=A(n,m,k)$, $B=B(n,m,k)$, and $Q=Q_{m,k}$.
Since $\int_0^q r(u)\,du = A + q B$, when $q \in (0,1]$ we have
$\phi(q) = ({Q(q)}/{q})(A+q B),$ and differentiating with respect to $q$ gives
\begin{align*}
  \phi'(q) = \frac{Q'(q)}{q}\big(A+qB\big) + \frac{Q(q)}{q}B - \frac{Q(q)}{q^2}\big(A+qB\big).
\end{align*}
Since $B = ({1-nA})/{k}$, in the previous expression we get
\begin{align*}
  \phi'(q) &= \frac{Q'(q)}{q}\left(A\left(1 - q \frac{n}{k}\right) + \frac{q}{k}\right) + \frac{Q(q)}{q}\frac{1-nA}{k} - \frac{Q(q)}{q^2}\left(A\left(1 - q \frac{n}{k}\right) + \frac{q}{k}\right)\\
  &= \frac{Q'(q)}{q}\left(A\left(1 - q \frac{n}{k}\right) + \frac{q}{k}\right) - \frac{Q(q) A}{q^2},
\end{align*}
which evaluated at $q = k/n$ gives
$\phi'(k/n) = {Q'(k/n)}/{k} - Q(k/n)  ({n}/{k})^2 A=0,$
where the last equality follows from the definition of $A$.
Therefore, $q=k/n$ is a critical point of the function $\phi$. 
To prove that $k/n$ is the unique global maximum of $\phi$, we first split the proof into the cases of $m = k$, $m = k + 1$ and then for $m \geq k + 2$. In the two first cases, we got that $\phi$ is simpler, and for the other case we reformulate $\phi'(q)$ in terms of an auxiliary function $R(q)$, and then reduce the problem to studying the function $R$ instead.
Note that for $m = k$, we have that $Q(q) = kq$ since $Y \leq k$. Then $B = 1/k -Q'(k/n)/(nQ(k/n)) = 0$ and so $\phi$ is constant. For $m = k+1$, we have
\[
    Q(q) = \EE[\min\{k,Y\}] = \EE[Y - \mathds{1}_{\{k+1\}}(Y)] = (k+1)q - q^{k+1}, \quad \text{as }Y\sim \text{Binomial}(k+1,q).
\]
Then, $\phi(q) = (k+1-q^k)(A+Bq)$ and so $\phi'(q) = B(k+1-q^k) - kq^{k-1}(A+Bq)$, which is decreasing. By the previous analysis, $\hat{q}=k/n$ is a critical point and then it is the maximum. For the general case, we prove it through the next propositions assuming that $m \geq k+2$.
\begin{proposition}\label{lem:reformulation}
For every $q\in [0,1]$, we have $
\phi'(q) = B(Q(q) - Q'(q)q)( R(q) - R(k/n))/q^2,$
where $R(q) = {q^2 Q'(q)}/{(Q(q) - qQ'(q))}.$
In particular, the sign of $\phi'(q)$ is equal to the sign of $R(q) - R(k/n)$.
\end{proposition}
\begin{proof}
We compute $\phi'(q)$ directly to get
\begin{align*}
\phi'(q) 
&= \frac{A(Q'(q)q - Q(q)) + BQ'(q)q^2}{q^2}.
\end{align*}
Since $\phi'(\hat{q}) = 0$ for $\hat{q} = k/n$, we get
$A = -{B\hat{q}^2Q'(\hat{q})}/{(Q'(\hat{q})\hat{q} - Q(\hat{q}))}.$
By substituting into $\phi'(q)$,
\begin{align*}
\phi'(q) &= \frac{1}{q^2}\left(-\frac{B\hat{q}^2Q'(\hat{q})}{Q'(\hat{q})\hat{q} - Q(\hat{q})}(Q'(q)q - Q(q)) + BQ'(q)q^2\right)\\
&= \frac{B(Q(q) - Q'(q)q)}{q^2} \left( \frac{q^2Q'(q)}{Q(q) - Q'(q)q} - \frac{\hat{q}^2Q'(\hat{q})}{Q(\hat{q}) - Q'(\hat{q})\hat{q}} \right)= \frac{B(Q(q)-Q'(q)q)}{q^2} (R(q) - R(\hat{q})),
\end{align*}
which proves the first part. 
Since $B > 0$ and $Q(q) - Q'(q)q\ge 0$ for every $q\in [0,1]$ (by $Q$ concave), the sign of $\phi'(q)$ is determined by $R(q) - R(k/n)$.
\end{proof}

\begin{proposition}\label{prop:key-reduction}
If $R$ is strictly decreasing on $(0,1)$, then the function $\phi$ achieves its unique global maximum at the point $q= k/n$.
\end{proposition}
\begin{proof}
By Proposition~\ref{lem:reformulation}, $\phi'(q)$ has the same sign as $R(q) - R(k/n)$. 
If $R$ is strictly decreasing, then for $q < k/n$ we have $R(q) > R(k/n)$, so $\phi'(q) > 0$, and for $q > k/n$ we have $R(q) < R(k/n)$, so $\phi'(q) < 0$. Thus $\phi$ has a unique global maximum at $q = k/n$.
\end{proof}
Therefore, by the previous propositions, to conclude Lemma \ref{prop:supremum-q-d} it is sufficient to show that $R$ is strictly decreasing on $(0,1)$.
To this end, it will be convenient to use integral expressions for $Q'$ and $Q-qQ'$.
We will use that for all $q \in (0,1)$ we have
  \begin{align}
    Q'(q) &= c \int_q^1 w(t)\,dt = c \cdot I(q),\label{integral1}\\
    Q(q) - qQ'(q) &= c \int_0^q t\,w(t)\,dt = c \cdot J(q),\label{integral2}
  \end{align}
  where $c = m(m - 1)$ and $w(q) = \binom{m - 2}{k-1} q^{k-1}(1-q)^{m - k - 1}$;
  in particular, $R(q) = {q^2 I(q)}/{J(q)}$.
  Indeed, to check \eqref{integral1}, observe that
  $Q'(q)=Q'(1)+ \int_{1}^q Q''(t)\,dt=c \int_q^1 w(t)\,dt = c \cdot I(q),$
  where the second equality holds by $Q''(t) = -c w(t)$ (Proposition \ref{prop:expected-value-algorithm}\ref{characterization-sol-c}) and $Q'(1) = 0$ since $m - k - 1 > 0$.
  Similarly, to check \eqref{integral2}, observe that
  \[
  \frac{d}{dq}(Q(q) - qQ'(q)) = Q'(q) - Q'(q) - qQ''(q) = -qQ''(q) = c q w(q),
  \]
  and therefore,
  $Q(q) - qQ'(q) = c \int_0^q t\,w(t)\,dt = c \cdot J(q).$
By computing the derivative of $R(q)$, we get
\begin{align*}
R'(q) &= \frac{1}{J(q)^2}\Big((2qI(q) + q^2 I'(q))J(q) - q^2 I(q) J'(q)\Big)\\
&=\frac{q}{J(q)^2}\Big(2I(q)J(q) - q w(q) J(q) - q^2 w(q) I(q)\Big),
\end{align*}
where we used that $I'(q) = -w(q)$ and $J'(q) = qw(q)$.

Define $h(q) = 2I(q)J(q) - qw(q)J(q) - q^2 w(q) I(q)$. 
Since $q > 0$ and $J(q) > 0$, we conclude that $R'(q) < 0$ if and only if $h(q) < 0$.
Before concluding Lemma \ref{prop:supremum-q-d}, we need one more technical proposition regarding $h$.
\begin{proposition}\label{lem:critical-point}
If $h'(\hat{q}) = 0$ for some $\hat{q} \in (0,1)$, then $h(\hat{q}) < 0$.
\end{proposition}
Before showing the proof of Proposition \ref{lem:critical-point}, we show how to conclude the lemma.
\begin{proof}[Proof of Lemma \ref{prop:supremum-q-d}]
First, note that by the previous analysis, we checked that $k/n$ is the maximum for $m = k$ and $m = k+1$. For $m \geq k+2$ we will show that $h(q) < 0$ for all $q \in (0,1)$.
We have $h(0) = 0$ since $J(0) = 0$ and $w(0) = 0$. Similarly, $h(1) = 0$ since $I(1) = 0$ and $w(1) = 0$ since $m - k - 1 > 0$. If $h(q_0) > 0$ for some $q_0 \in (0,1)$, then by continuity and differentiability $h$ achieves a positive local maximum at some interior point $\hat{q}$ where $h'(\hat{q}) = 0$ and $h(\hat{q}) > 0$. 
However, Proposition~\ref{lem:critical-point} shows $h(\hat{q}) < 0$, which is a contradiction. 
Since $h(q) < 0$ for all $q \in (0,1)$, we have $R'(q) < 0$, so $R(q)$ is strictly decreasing. 
Proposition \ref{prop:key-reduction} implies that $\phi$ achieves its unique global maximum at $k/n$.
\end{proof}

\begin{proof}[Proof of Proposition \ref{lem:critical-point}]
Computing $h'(q)$ and using that $I'(q) = -w(q)$ and $J'(q) = qw(q)$, we have
\begin{align*}
h'(q) &= 2I'(q)J(q) + 2I(q)J'(q) - w(q)J(q) - qw'(q)J(q) - qw(q)J'(q) \\
&\quad - 2qw(q)I(q) - q^2 w'(q)I(q) - q^2 w(q)I'(q)= -3w(q)J(q) - qw'(q)J(q) - q^2 w'(q)I(q).
\end{align*}
Let $\ell(q) = {w'(q)}/{w(q)} = {(k-1)}/{q} - {(m - k - 1)}/{(1-q)}$. Then
\[
h'(q) = -w(q)\big(3J(q) + q\ell(q)J(q) + q^2 \ell(q)I(q)\big).
\]
Let $\hat{q} \in (0,1)$ be a critical point where $h'(\hat{q}) = 0$. Since $w(\hat{q}) \neq 0$, we have
\begin{equation}\label{eq:critical-cond}
J(\hat{q})(3 + \hat{q}\ell(\hat{q})) = -\hat{q}^2 \ell(\hat{q}) I(\hat{q}).
\end{equation}
We first show that $\ell(\hat{q}) < 0$. 
Suppose, for the sake of contradiction, that $\ell(\hat{q}) \geq 0$.
Then the right-hand side of \eqref{eq:critical-cond} is less or equal to 0, while the left-hand side satisfies $J(\hat{q})(3 + \hat{q}\ell(\hat{q})) > 0$ since $J(\hat{q}) > 0$ and $3 + \hat{q}\ell(\hat{q}) \geq 3$. 
This shows $\ell(\hat{q}) < 0$. Therefore, since $w(q) > 0$ for all $q \in (0,1)$, we have that $w'(\hat{q}) < 0$ because $\ell(q) = w'(q)/w(q)$. 
The function $w(q)$ is a beta density, and therefore, unimodal with mode at $(k-1)/(m - 2) \in (0,1)$, since $m-k-1 > 0$. 
In particular, we have $\hat{q} \geq {(k-1)}/{(m - 2)}$ because only at the right of the mode the beta-type function is decreasing.
From the equality \eqref{eq:critical-cond} we have $J(\hat{q}) = {\hat{q}^2 |\ell(\hat{q})| I(\hat{q})}/{(3 + \hat{q}\ell(\hat{q}))}$, and thus, by substituting into $h(\hat{q})$ we get
\begin{align}
h(\hat{q}) = J(\hat{q})(2I(\hat{q}) - \hat{q}w(\hat{q})) - \hat{q}^2 w(\hat{q}) I(\hat{q}) = \frac{-\hat{q}^2 I(\hat{q})}{3 + \hat{q}\ell(\hat{q})}(2\ell(\hat{q})I(\hat{q}) + 3w(\hat{q})).\label{reduction-h}
\end{align} 
In particular, we have $h(\hat{q}) < 0$ if and only if $3w(\hat{q}) + 2\ell(\hat{q})I(\hat{q}) > 0$. 
To verify this inequality, we compute $I(q)$ explicitly via integration by parts. For $r\in \{0, 1, \ldots, k-1\}$, define $I_r(q) = \int_q^1 w_r(s)\,ds$ where $w_r(s) = \binom{m - 2}{k-1} s^{k-1-r}(1-s)^{m - k - 1 + r}$, and $I_0(q) = I(q)$. 
Integrating by parts gives
\[
I_r(q) = \binom{m - 2}{k-1} \frac{q^{k-1-r}(1-q)^{m - k + r}}{m - k + r} + \frac{k-1-r}{m - k + r} I_{r+1}(q).
\]
Note that for $r = k-1$, we have that:
\begin{align*}
I_{k-1}(q) &= \int_q^1 \binom{m - 2}{k-1}t^{k-1-k+1}(1-t)^{m - k - 1 + k - 1}\,dt= \binom{m - 2}{k-1} \frac{(1-q)^{m - 1}}{m - 1}.
\end{align*}
Iterating the recurrence for $I_r(q)$ (i.e., unrolling it up to $r=k-1$) yields
\[
I(q)=\binom{m-2}{k-1}\sum_{j=0}^{k-1}\frac{(k-1)_j}{(m-k)^{(j+1)}}\,q^{k-1-j}(1-q)^{m-k+j}
= w(q)(1-q)\sum_{j=0}^{k-1} a_j\,x(q)^j,
\]
where $a_j=\frac{(k-1)_j}{(m-k)^{(j+1)}}$, $x(q)=\frac{1-q}{q}$ with $(k-1)_j = (k-1)(k-2)\cdots(k-j)$ for $j > 0$, $(k-1)_0 = 1$, and $(m - k)^{(j+1)} = (m - k)\cdots(m - k + j)$ for all $j \geq 0$. 
Then, from the equality in \eqref{reduction-h}, it is sufficient to show that $2\ell(\hat{q})w(\hat{q})(1-\hat{q}) \sum_{j=0}^{k-1} a_j x(q)^j + 3w(\hat{q}) > 0$, or equivalently,
\begin{equation}\label{eq:ineq-label}
2\ell(\hat{q})(1-\hat{q}) \sum_{j=0}^{k-1} a_j x(\hat{q})^j + 3 > 0.
\end{equation}
Since $\hat{q} \ge {(k-1)}/{(m - 2)}$, we have
$x(\hat{q}) = {(1-\hat{q})}/{\hat{q}} \le  {(m - k - 1)}/{(k-1)}.$
Also, note that we can bound $a_j$ as follows:
\[
a_j = \frac{(k-1)_j}{(m - k)^{(j+1)}} = \frac{(k-1)(k-2)\cdots(k-j)}{(m - k)\cdots(m - k + j)} \leq \frac{(k-1)^j}{(m - k)^{j+1}}.
\]
Since $(k-1)x(\hat{q}) < m - k - 1 < m - k$ we have ${(k-1)x(\hat{q})}/{(m - k)} < 1$ and $m - k -1 - (k-1)x(\hat{q}) > 0$. Therefore, we have
\[
\sum_{j=0}^{k-1} a_j x(\hat{q})^j \leq \frac{1}{m - k} \sum_{j=0}^{\infty} \left[\frac{(k-1)x(\hat{q})}{m - k}\right]^j = \frac{1}{(m - k) - (k-1)x(\hat{q})} \leq \frac{1}{m - k - 1 - (k-1)x(\hat{q})}.
\]
On the other hand, we have $\ell(\hat{q})(1-\hat{q})=(k-1)x(\hat{q}) - (m - k - 1)$,
so we can conclude \eqref{eq:ineq-label} as follows knowing that $m - k \geq m - k - 1 - (k-1)x(\hat{q}) \geq 0$:
\begin{align*}
2\ell(\hat{q})(1-\hat{q}) \sum_{j=0}^{k-1} a_j x(\hat{q})^j + 3 &= 2((k-1)x(\hat{q}) - (m - k - 1)) \sum_{j=0}^{k-1} a_j x(\hat{q})^j + 3\\
&\geq -2 \frac{(m - k - 1 - (k-1)x(\hat{q}))}{m - k - 1 - (k-1)x(\hat{q})} + 3\geq -2 + 3 = 1,
\end{align*}
i.e., $h(\hat{q}) < 0$. 
This concludes the proof.
\end{proof}

\section{Competition Complexity Characterization and Bounds}\label{sec:CCbounds}

In this section, we use Theorem~\ref{thm:optimal-solution-PNK} to characterize the $(1-\varepsilon,k)$-competition complexity $\beta_k(\varepsilon)$ (Theorem~\ref{thm:beta_characterization}). Then, we provide bounds on $\beta_k(\varepsilon)$ as a function of $k$ and $\varepsilon$ (Theorem~\ref{thm:betak-epsilon-bounds}).
Recall that $\beta_k(\varepsilon) = \limsup_{n\to \infty} \beta_{k,n}(\varepsilon)$, where $\beta_{k,n}(\varepsilon)$ is the infimum ratio $m/n$ (with $m \in \NN$, $m \geq k$) for which the $(m,n,k)$-competitive ratio is at least $1-\varepsilon$. 
Our first result is the following characterization which follows by taking the limit in Binomial distributions.

{\begin{theorem}\label{thm:beta_characterization}
    For any $\varepsilon\in (0,1)$, and $k\geq 1$, the $(1-\varepsilon,k)$-competition complexity $\beta_{k}(\varepsilon)$ satisfies 
    \[
   \mathbf{E}[\min\{ k, \mathrm{Poisson}(\beta_k(\varepsilon) \cdot k)  \}]=(1-\varepsilon)k.
    \]
\end{theorem}

\begin{proof}
    From the definition of $\beta_k(\varepsilon)$, we have
    \begin{align}
        \beta_{k}(\varepsilon)         & = \limsup_{n\geq k} \inf_{m\geq k} \{  m/n : Q_{m,k}(k/n)/k\geq 1-\varepsilon  \}\nonumber\\
        & = \lim_{n'\to \infty} \sup_{n\geq n'}\sup_{m\geq k} \{ m/n : Q_{m-1,k}(k/n)/k< 1-\varepsilon \}\label{eq:new_form_beta}
    \end{align}
    where in the first line we used Theorem~\ref{thm:optimal-solution-PNK} together with the definition of $\beta_{n,k}(\varepsilon)$, and in the last line we used the monotonicity of $Q_{m,k}$ in $m$ to rewrite the infimum as a supremum. With the last characterization, we first show that $\beta_k(\varepsilon)$ equals $\beta^*=\sup\{ \beta : \mathbf{E}[\min\{ k, \mathrm{Poisson}(\beta \cdot k) \} ]/k < 1-\varepsilon \}$. The conclusion of the theorem then will follow because the function $P_k(\lambda)=\mathbf{E}[\min\{ k,\mathrm{Poisson}(\lambda) \}]$ is increasing and continuous in $\lambda$. Note that because of the continuity of $P_k(\cdot)$, we can use $\leq$ instead of $<$ in the definition of $\beta^*$, i.e., $\beta^*=\sup\{ \beta: P_k(\beta k)\leq k(1-\varepsilon) \}$.

    First, we prove that $\beta^* \leq \beta_k(\varepsilon)$. Let $\beta>0$ such that $P_{k}(\beta k)< (1-\varepsilon)k$. Now, consider $n\geq 1$ and $m=m(n)$ such that $m/n \leq \beta$ and $m /n \to \beta$. Le Cam's theorem~\cite{le1960approximation,steele1994cam}, guarantees that $\left|\PP[\mathrm{Poisson}(m k/n)\in S]- \PP[\mathrm{Binom}(m,k/n)\in S]\right| \leq 2 m k^2/n^2$ for any $S\subseteq \NN$ and so
    \begin{align*}
        Q_{m-1,k}(k/n) & \leq Q_{m,k}(k/n)  \leq  P_k(mk/n) + 2\frac{ mk^3}{n^2} \leq P_{k}(\beta k) + 2 \beta \frac{k^3}{n}.
    \end{align*}    
    where in the first line, we used the monotonicity of $Q_{\cdot,k}(k/n)$, in the second inequality we used the Le Cam's theorem and in the last inequality we used the monotonicity of $P_k(\cdot)$. Since $P_k(\beta k) < (1-\varepsilon) k $ then, there is $n_0\geq k$ and $m\geq k$ such that $Q_{m-1,k}(k/n) < (1-\varepsilon)k$ for any $n\geq n_0$. Thus, $m/n \leq \beta_k(\varepsilon)$ for any $n\geq n_0$ and by taking limit in $n$ we obtain $\beta \leq \beta_k(\varepsilon)$. Since $\beta$ was arbitrary such that $P_k(\beta k) <(1-\varepsilon)k$, we also obtain $\beta^* \leq \beta_k(\varepsilon)$.

    We now prove $\beta_k(\varepsilon)\leq \beta^*$. Fix $\delta\in (0,\beta_k(\varepsilon)/2)$ and take $n_0$ large enough such that there is $n\geq n_0$ and $m\geq k$ that satify $|m/n - \beta_k(\varepsilon)|\leq \delta$ and $Q_{m-1}(k/n) < 1-\varepsilon$. If we increase $n_0$ enough, we can also ensure $1/n \leq \delta$ and so $(m-1)/n \geq \beta_k(\varepsilon) - 2\delta$. Using Le Cam's theorem, we have
    \begin{align*}
        P_{k}((\beta_k(\varepsilon) - 2 \delta) k) & \leq P_{k}( (m-1)k/n ) \\
        & \leq Q_{m-1,k}(k/n) + 2(\beta_k(\varepsilon)+\delta) k^3/n_0 < (1-\varepsilon)k + 2(\beta_k(\varepsilon)+\delta)k^3/n_0.
    \end{align*}
    By taking $n_0\to \infty$, we obtain $P_{k}((\beta_k(\varepsilon)-2\delta) k)\leq (1-\varepsilon)k$ which implies $\beta_k(\varepsilon)-2\delta \leq \beta^*$. Since $\delta$ is arbitrary, we conclude that $\beta_k(\varepsilon)\leq \beta^*$ and this finishes the proof.
\end{proof}
}

Our next result provides analytical bounds on $\beta_k(\varepsilon)$ as a function of $k$ and $\varepsilon$.

\begin{theorem}\label{thm:betak-epsilon-bounds}
For $k=1$ we have $\beta_k(\varepsilon)=\ln(1/\varepsilon)$. For $k > 1$ we have 
\[
\frac{\ln (1/\varepsilon)}{k} \leq \beta_k(\varepsilon)\leq \left(1 - \frac{1}{k}\right) \left( 1 + \frac{\ln (1/\varepsilon)}{k-1} + \sqrt{ \left( \frac{\ln(1/\varepsilon)}{k-1} \right)^2 + 2\frac{\ln(1/\varepsilon)}{k-1}} \right).
\]
\end{theorem}

Using the subadditivity of the square root function, we can prove 
$$\beta_k(\varepsilon)\leq  1 +  \frac{2\ln({1}/{\varepsilon} )}{k-1} + \sqrt{\frac{2\ln({1}/{\varepsilon})}{k-1}} .$$ This bound implies the existence of a constant $\gamma>0$ such that $\beta_k(\exp(-\gamma k))\le 1.01$, i.e., an extra $1\%$ increase over $n$ ensures that the optimal single-threshold algorithm on a sequence of length $m$ secures a $1-\exp{(-\gamma k)}$ fraction of the prophet’s value---in stark contrast with the case $m=n$, where the best known bound is $\approx 1-1/\sqrt{2\pi k}$.
Naturally, we can replace the 1\% by any constant at the cost of changing the constant $\gamma$.  
The rest of this section is devoted to proving Theorem~\ref{thm:betak-epsilon-bounds}.

Recall that by Theorem \ref{thm:optimal-solution-PNK} and Proposition \ref{prop:expected-value-algorithm}\ref{characterization-sol-a}, the $(m,n,k)$-competitive ratio is achieved by the single-threshold algorithm with quantile $q = k/n$, and its value is equal to $Q_{m,k}(k/n)/k=(k-\EE[(k-Y)_+])/k$, with $Y\sim \normalfont\text{Binomial}(m,k/n)$. Therefore, the competitive ratio bound is equivalent to $\left( k - \EE[(k - Y)_+] \right)/k \geq 1 - \varepsilon$ for all $n \geq k$, where $Y \sim \normalfont\text{Binomial}(m, k/n)$, that is,
  $\EE[(k - Y)_+] \leq k \varepsilon$ for all $n\geq k$.
We start by proving the following proposition.
\begin{proposition}\label{prop:beta-k-growth}
For every $\varepsilon \in (0,1)$ and $k \in \NN$, we have
  \[
    \frac{\ln(1/\varepsilon)}{k} \leq \beta_k(\varepsilon) \leq \inf_{t>0} \psi_k(t, \varepsilon),
  \]
  where $\psi_k(t,\varepsilon) = (t(k-1) + \ln\left({1}/{\varepsilon} \right))/(k(1-e^{-t}))$. 
\end{proposition}
\begin{proof}
  First, we prove the upper bound.
  Consider $\varepsilon \in (0,1)$, $k \in \NN$ and $m,n \in \NN$ such that the $(m,n,k)$-competition complexity is at least $1-\varepsilon$.
  For the random variable $Y \sim \normalfont\text{Binomial}(m, k/n)$, note that $(k - Y)_+ \leq k \mathds{1}_{[0,k-1]}(Y)$, which implies that $\EE[(k - Y)_+] \leq k\cdot \PP[Y \leq k-1].$
  On the other hand, for any $t > 0$, we have
  \begin{align*}
    \PP[Y \leq k-1] &= \PP[\exp{(-tY)} \geq \exp{(-t(k-1))}]\\
    &\leq \exp{(t(k-1))} \EE[\exp{(-tY)}]\\
    &= \exp{(t(k-1))} \Big(1-\frac{k}{n}\Big(1 - e^{-t}\Big)\Big)^{m}\\
    &\leq \exp{(t(k-1))} \exp{\Big(-\frac{k}{n} \Big( 1 - e^{-t} \Big) m\Big)}= \exp\left( t(k-1) - k\frac{m}{n} \Big( 1 - e^{-t}\Big)  \right),
  \end{align*}
  where the first inequality holds by Markov's bound, and the second inequality by $1 - u \leq \exp(-u)$ for $u \geq 0$.
  Then, a sufficient condition for the inequality $\EE[(k - Y)_+] \leq k \varepsilon$ to be satisfied is that $t(k-1) -k ({m}/{n})(1 - \exp(-t)) \leq \ln(\varepsilon),$
  or equivalently,
    ${m}/{n} \geq  \psi_k(t,\varepsilon).$
  Since the bound works for any $t > 0$, we get $\beta_k(\varepsilon) \leq \inf_{t > 0} \psi_k(t,\varepsilon)$, concluding the upper bound. 
  For the lower bound, note that $(k - Y)_+ \geq k\mathds{1}_{\{0\}}(Y)$, and therefore,
  $\EE[(k - Y)_+] \geq k \cdot \PP[Y = 0] = k (1 - {k}/{n})^{m}$.
  Then, $\beta_k(\varepsilon)$ satisfies $\sup_{n \geq k} \left(1 - {k}/{n}\right)^{\beta_k(\varepsilon) n} \leq \varepsilon$.
  Since this sequence is increasing in $n$, the supremum is the limit in $n$, i.e., $\exp{(-\beta_k(\varepsilon) k)} \leq \varepsilon$.
  By rearranging we get $\beta_k(\varepsilon) \geq \ln({1}/{\varepsilon})/k$.
\end{proof}
In the following lemma, we solve the one-dimensional optimization problem in the upper bound of Proposition \ref{prop:beta-k-growth} when $k>1$. 
Namely, we find explicitly the best upper bound for $\beta_k(\varepsilon)$.
\begin{lemma}\label{lemma:minimum-t}
  For every $\varepsilon\in (0,1)$ and every $k>1$ we have
$\inf_{t > 0} \psi_k(t,\varepsilon) = \psi_k(t^*,\varepsilon),$
  where $t^*$ is the unique strictly positive value satisfying
  \begin{equation}\label{eq:t-star-eq}
    e^{t^*} = \frac{\ln({1}/{\varepsilon})}{k-1} + 1 + t^*.
  \end{equation}
  In particular, we have $t^* = \ln( - W_{-1}( - \exp({-\ln( {1}/{\varepsilon})}/{(k-1)} - 1)))$.\footnote{$W_{-1}$ is the second branch of the Lambert $W$ function.}
\end{lemma}
The proof follows by showing that $\psi_k(t,\varepsilon)$ is convex over the strictly positive reals for every $\varepsilon\in (0,1)$; then \eqref{eq:t-star-eq} is obtained from the first-order optimality condition. 

\begin{proof}[Proof of Lemma \ref{lemma:minimum-t}]
  First we prove that $\psi_k(t,\varepsilon)$ is convex in $t > 0$.
  Just for notation simplicity we work with $\varphi_k(t,\varepsilon) = k \psi_k(t;\varepsilon)$ instead. 
  Then, by denoting $a = k-1$ and $L = \ln\left({1}/{\varepsilon} \right)$, we have
  \begin{align*}
    \varphi_k'(t,\varepsilon) 
    &= \frac{a - ae^{-t} - ate^{-t} - e^{-t}L}{(1-e^{-t})^2},\\
    \varphi_k''(t;\varepsilon) 
    &= \frac{e^{2t}\left( at + L - 2a \right) + e^{t}\left( at + L + 2a \right)}{(e^t - 1)^3}.
  \end{align*}
  We define $x = e^t > 1$ as $t > 0$. Then, note that
    $\varphi_k''(t,\varepsilon) = \varphi_k''(\ln(x),\varepsilon)= \calJ(x)\cdot {x}/{(x-1)^3} ,$
  where $\calJ(x) = x\left( a\ln(x) + L - 2a \right) + a\ln(x) + L + 2a$. 
  Note that ${x}/{(x-1)^3} > 0$ for $x > 1$, and then we only have to prove that $\calJ(x) > 0$ for $x > 1$. 
  Observe that $\calJ(1) = L - 2a + L + 2a = 2\ln\left({1}/{\varepsilon} \right) > 0.$
  Now we show that $\calJ$ is increasing in $(1, \infty)$. 
  By simply taking the derivative, we have
  \begin{align*}
    \calJ'(x) = a\ln(x) + L - 2a + x\cdot \frac{a}{x} + a \cdot \frac{1}{x} = a \left( \ln(x) + \frac{1}{x} - 1 \right) + L \geq L > 0,
  \end{align*}
  where the inequality holds since $\ln(x) + {1}/{x} - 1 \geq 0$ for $x > 0$.
  Therefore, $\calJ(x) \ge \calJ(1)>0$ for $x > 1$, which implies that $\varphi_k''(t,\varepsilon) > 0$ for $t > 0$, concluding that $\varphi_k(t,\varepsilon)$ is convex in the strictly positive reals.
  On the other hand, $\psi_k'(t^*,\varepsilon) = 0$ if and only if
 $e^{t^*} = {L}/{a} + 1 + t^*$,
  i.e., \eqref{eq:t-star-eq} holds. 
  By standard algebraic manipulation of the previous equation, we arrive at the exact closed expression for $t^*$, that is, $t^* = \ln( - W_{-1}( - \exp({-\ln( {1}/{\varepsilon})}/{(k-1)} - 1)))$.
\end{proof}

We are now able to prove Theorem~\ref{thm:betak-epsilon-bounds}. 
\begin{proof}[Proof of Theorem \ref{thm:betak-epsilon-bounds}]
    For $k=1$, using Theorem~\ref{thm:beta_characterization}, we have $1-\varepsilon=\EE[\min\{1, \mathrm{Poisson}(\beta_1(\varepsilon)) \}] = 1-e^{-\beta_1(\varepsilon)} $ from where we immediately obtain $\beta_1(\varepsilon)=\ln(1/\varepsilon)$.
  For $k>1$, Proposition \ref{prop:beta-k-growth} and Lemma \ref{lemma:minimum-t} imply that $\beta_k(\varepsilon) \leq \psi_k(t^*,\varepsilon)$. 
  On the other hand, we have
  \begin{align*}
    \psi_k(t^*,\varepsilon) &= \frac{t^*(k-1) + \ln({1}/{\varepsilon})}{k(1-e^{-t^*})}\\
    &= \frac{(k-1)e^{t^*} - (k-1)}{k(1-e^{-t^*})}= (k-1) \frac{e^{t^*} - 1}{k(1-e^{-t^*})} =  \left(1 - \frac{1}{k}\right) e^{t^*},
  \end{align*}
  where the second equality follows from \eqref{eq:t-star-eq}.
  Also from \eqref{eq:t-star-eq}, we have $z - 1 - \ln(z) = \theta$, where $z=e^{t^*}$ and $\theta = {\ln({1}/{\varepsilon})}/{(k-1)}.$
  Then, we get
  \[
  \theta = z - 1 - \ln(z) = \int_1^z \frac{u-1}{u} du\ge \int_1^z \frac{u-1}{z} du=\frac{(z-1)^2}{2z}.
  \]
  Previous inequality is equivalent to
  $z^2 - 2( 1 + \theta)z + 1 \leq 0,$
  which is a quadratic inequality with positive coefficient, that is, $r_- \leq z \leq r_+$, where $r_{-}$ and $r_+$ are the positive and negative roots of the quadratic equation, respectively. 
  The roots are $r_{\pm} = 1 + \theta \pm \sqrt{\theta^2 + 2\theta}$. 
  For the positive root, we have $z \leq 1 + \theta + \sqrt{\theta^2 + 2\theta}$ and therefore we get the following bound:
  \begin{align*}
  \beta_k(\varepsilon) &\leq  \left(1 - \frac{1}{k}\right) e^{t^*}\leq \left(1 - \frac{1}{k}\right) \left( 1 + \theta + \sqrt{\theta^2 + 2\theta} \right)
  \end{align*}
  This prove the upper bound for $k>1$.
  The lower bound on $\beta_k(\varepsilon)$ comes directly from Proposition \ref{prop:beta-k-growth}.
  This concludes the proof.
\end{proof}

\section{Final Remarks}\label{sec:conclusion}

We initiated the study of competition complexity for the multiple-selection prophet inequality problem and provided an exact characterization of the $(1-\varepsilon,k)$-competition complexity for single-threshold algorithms. Our results have immediate implications for posted-price mechanisms, in particular for fixed-price mechanisms in multi-item auctions. In line with the insight of Bulow and Klemperer, our findings suggest that increasing market size can be more beneficial than pursuing more sophisticated or optimal mechanisms. 

In this work, we focused exclusively on single-threshold algorithms (i.e., static posted-price mechanisms), which cleanly illustrate the contrast between simplicity and complexity. A natural open question is to characterize the competition complexity of the (more complex) optimal online multi-threshold algorithm. We remark that, even in the standard multiple-selection i.i.d.\ prophet inequality problem, characterizing the competitive ratio of the optimal algorithm remains open despite several efforts (see, e.g.,~\citep{jiang2025tightness,brustle2025splitting}). 
Furthermore, even when it might be possible to use the approach of~\cite{brustle2025splitting} to derive a nonlinear system of ordinary differential equations that yields a lower bound on the ratio $\DP_{m,k}(F)/\OPT_{n,k}(F)$ for large $n$ and $m \ge n$, obtaining analytical bounds for $k>1$ from this nonlinear system has remained elusive, even in the absence of competition.

Beyond the i.i.d.\ model, it is natural to ask about the competition complexity of general independent---but not necessarily identical---multiple-selection prophet inequality problems. For example, in the adversarial-arrival setting, the competitive ratio is known to be at least $1 - 1/\sqrt{k+3}$ (see~\cite{alaei2014bayesian}), while for the optimal single-threshold algorithm it is known to be $1 - \mathcal{O}(\sqrt{\ln(k)/k})$ (see, e.g.,~\cite{chawla2024static,hajiaghayi2007automated}). In such models, it is natural to pursue competition complexity in the block model~\cite{brustle2025competition}, where the online decision-maker observes $\ell \ge 1$ independent copies of a sequence of length $n$, i.e., $m = \ell n$ with $\ell$ a nonnegative integer, as opposed to the notion of competition complexity studied in this work. Addressing these settings will likely require new techniques, as our approach relies heavily on the i.i.d.\ assumption. Similar questions may also be explored in the random-order model, i.e., the prophet secretary problem~\cite{arnosti2023tight}.

\bibliographystyle{plain}
{\small\bibliography{referencesOR}}

@article{steele1994cam,
  title={Le Cam's inequality and Poisson approximations},
  author={Steele, J Michael},
  journal={The American Mathematical Monthly},
  volume={101},
  number={1},
  pages={48--54},
  year={1994},
  publisher={Taylor \& Francis}
}

@article{le1960approximation,
  title={An approximation theorem for the Poisson binomial distribution.},
  author={Le Cam, Lucien},
  journal={Pacific J. Math.},
  volume={10},
  number={4},
  pages={1181--1197},
  year={1960}
}

@book{devroye1986non,
  title={Non-Uniform Random Variate Generation},
  author={Devroye, Luc},
  year={1986},
  publisher={Springer-Verlag},
  address={New York}
}

@article{correa2021posted,
  title={Posted price mechanisms and optimal threshold strategies for random arrivals},
  author={Correa, Jos{\'e} and Foncea, Patricio and Hoeksma, Ruben and Oosterwijk, Tim and Vredeveld, Tjark},
  journal={Mathematics of Operations Research},
  volume={46},
  number={4},
  pages={1452--1478},
  year={2021},
  publisher={INFORMS}
}

@article{hill1982comparisons,
  title={Comparisons of stop rule and supremum expectations of iid random variables},
  author={Hill, Theodore P and Kertz, Robert P},
  journal={The Annals of Probability},
  pages={336--345},
  year={1982},
  publisher={JSTOR}
}

@article{brustle2024competition,
  title={The Competition Complexity of Dynamic Pricing},
  author={Brustle, Johannes and Correa, Jos{\'e} and Duetting, Paul and Verdugo, Victor},
  journal={Mathematics of Operations Research},
  volume={49},
  number={3},
  pages={1986--2008},
  year={2024},
  publisher={INFORMS}
}

@article{brustle2025competition,
  title={The Competition Complexity of Prophet Inequalities},
  author={Brustle, Johannes and Correa, Jos{\'e} and D{\"u}tting, Paul and Ezra, Tomer and Feldman, Michal and Verdugo, Victor},
  journal={Mathematics of Operations Research},
  year={2025},
  publisher={INFORMS}
}

@inproceedings{chawla2010multi,
  title={Multi-parameter mechanism design and sequential posted pricing},
  author={Chawla, Shuchi and Hartline, Jason D and Malec, David L and Sivan, Balasubramanian},
  booktitle={Proceedings of the 42nd ACM Symposium on Theory of Computing (STOC)},
  pages={311--320},
  year={2010}
}

@article{correa2019pricing,
  title={From pricing to prophets, and back!},
  author={Correa, Jos{\'e} and Foncea, Patricio and Pizarro, Dana and Verdugo, Victor},
  journal={Operations Research Letters},
  volume={47},
  number={1},
  pages={25--29},
  year={2019},
  publisher={Elsevier}
}

@article{arnosti2023tight,
  title={Tight guarantees for static threshold policies in the prophet secretary problem},
  author={Arnosti, Nick and Ma, Will},
  journal={Operations Research},
  volume={71},
  number={5},
  pages={1777--1788},
  year={2023},
  publisher={INFORMS}
}

@article{perez2025iid,
  title={The iid prophet inequality with limited flexibility},
  author={Perez-Salazar, Sebastian and Singh, Mohit and Toriello, Alejandro},
  journal={Mathematics of Operations Research},
  year={2025},
  publisher={INFORMS}
}

@article{chawla2024static,
  title={Static pricing for multi-unit prophet inequalities},
  author={Chawla, Shuchi and Devanur, Nikhil and Lykouris, Thodoris},
  journal={Operations Research},
  volume={72},
  number={4},
  pages={1388--1399},
  year={2024},
  publisher={Informs}
}

@article{bulow1996auctions,
  title={Auctions versus negotiations},
  author={Bulow, Jeremy and Klemperer, Paul},
  journal={The American Economic Review},
  volume={86},
  number={1},
  pages={180},
  year={1996},
  publisher={American Economic Association}
}

@article{krengel1977semiamarts,
  title={Semiamarts and finite values},
  author={Krengel, Ulrich and Sucheston, Louis},
  journal={Bulletin of the American Mathematical Society},
  volume={83},
  number={4},
  pages={745--747},
  year={1977}
}

@article{correa2019recent,
  title={Recent developments in prophet inequalities},
  author={Correa, Jose and Foncea, Patricio and Hoeksma, Ruben and Oosterwijk, Tim and Vredeveld, Tjark},
  journal={ACM SIGecom Exchanges},
  volume={17},
  number={1},
  pages={61--70},
  year={2019},
  publisher={ACM New York, NY, USA}
}

@article{brustle2025splitting,
  title={Splitting guarantees for prophet inequalities via nonlinear systems},
  author={Brustle, Johannes and Perez-Salazar, Sebastian and Verdugo, Victor},
  journal={Mathematics of Operations Research},
  year={2025},
  publisher={INFORMS}
}

@article{jiang2025tightness,
  title={Tightness without counterexamples: A new approach and new results for prophet inequalities},
  author={Jiang, Jiashuo and Ma, Will and Zhang, Jiawei},
  journal={Mathematics of Operations Research},
  year={2025},
  publisher={INFORMS}
}

@inproceedings{perez2024optimal,
    author = {Perez-Salazar, Sebastian and Verdugo, Victor},
    title = {Optimal Guarantees for Online Selection Over Time},
    booktitle = {Proceedings of the Conference on Web and Internet Economics (WINE)},
    year = 2024
}

@article{wittmann1996superprophet,
  title={Superprophet inequalities for independent random variables},
  author={Wittmann, Rainer},
  journal={Journal of Applied Probability},
  volume={33},
  number={3},
  pages={904--908},
  year={1996},
  publisher={Cambridge University Press}
}

@article{samuel1984comparison,
  title={Comparison of threshold stop rules and maximum for independent nonnegative random variables},
  author={Samuel-Cahn, Ester},
  journal={The Annals of Probability},
  pages={1213--1216},
  year={1984},
  publisher={JSTOR}
}

@inproceedings{hartline2009simple,
  title={Simple versus optimal mechanisms},
  author={Hartline, Jason D and Roughgarden, Tim},
  booktitle={Proceedings of the ACM Conference on Electronic Commerce (EC)},
  pages={225--234},
  year={2009}
}

@inproceedings{feldman201899,
  title={99\% revenue via enhanced competition},
  author={Feldman, Michal and Friedler, Ophir and Rubinstein, Aviad},
  booktitle={Proceedings of the ACM Conference on Economics and Computation (EC)},
  pages={443--460},
  year={2018}
}

@inproceedings{eden2016competition,
  title={The competition complexity of auctions: A bulow-klemperer result for multi-dimensional bidders},
  author={Eden, Alon and Feldman, Michal and Friedler, Ophir and Talgam-Cohen, Inbal and Weinberg, S Matthew},
  booktitle={Proceedings of the ACM Conference on Economics and Computation (EC)},
  pages={343},
  year={2017}
}

@inproceedings{beyhaghi2019optimal,
  title={Optimal (and benchmark-optimal) competition complexity for additive buyers over independent items},
  author={Beyhaghi, Hedyeh and Weinberg, S Matthew},
  booktitle={Proceedings of the 51st ACM Symposium on Theory of Computing (STOC)},
  pages={686--696},
  year={2019}
}

@inproceedings{ezra2025competition,
  title={The competition complexity of prophet inequalities with correlations},
  author={Ezra, Tomer and Garbuz, Tamar},
  booktitle={Proceedings of the ACM Conference on Economics and Computation (EC)},
  pages={226--249},
  year={2025}
}

@article{beyhaghi2021improved,
  title={Improved revenue bounds for posted-price and second-price mechanisms},
  author={Beyhaghi, Hedyeh and Golrezaei, Negin and Leme, Renato Paes and P{\'a}l, Martin and Sivan, Balasubramanian},
  journal={Operations Research},
  volume={69},
  number={6},
  pages={1805--1822},
  year={2021},
  publisher={INFORMS}
}

@inproceedings{chakraborty2010approximation,
  title={Approximation schemes for sequential posted pricing in multi-unit auctions},
  author={Chakraborty, Tanmoy and Even-Dar, Eyal and Guha, Sudipto and Mansour, Yishay and Muthukrishnan, S},
  booktitle={Proceedings of the International Conference on Internet and Network Economics (WINE)},
  pages={158--169},
  year={2010}
}

@article{kertz1986stop,
  title={Stop rule and supremum expectations of iid random variables: a complete comparison by conjugate duality},
  author={Kertz, Robert P},
  journal={Journal of multivariate analysis},
  volume={19},
  number={1},
  pages={88--112},
  year={1986},
  publisher={Elsevier}
}

@article{kleinberg2019matroid,
  title={Matroid prophet inequalities and applications to multi-dimensional mechanism design},
  author={Kleinberg, Robert and Weinberg, S Matthew},
  journal={Games and Economic Behavior},
  volume={113},
  pages={97--115},
  year={2019},
  publisher={Elsevier}
}

@inproceedings{ezra2024competition,
  title={The competition complexity of prophet secretary},
  author={Ezra, Tomer and Garbuz, Tamar},
  booktitle={Proceedings of the ACM Conference on Economics and Computation (EC)},
  year={2026}
}

@article{alaei2014bayesian,
  title={Bayesian combinatorial auctions: Expanding single buyer mechanisms to many buyers},
  author={Alaei, Saeed},
  journal={SIAM Journal on Computing},
  volume={43},
  number={2},
  pages={930--972},
  year={2014},
  publisher={SIAM}
}

@inproceedings{hajiaghayi2007automated,
  title={Automated online mechanism design and prophet inequalities},
  author={Hajiaghayi, Mohammad Taghi and Kleinberg, Robert and Sandholm, Tuomas},
  booktitle={Annual Conference on Artificial Intelligence (AAAI)},
  volume={7},
  pages={58--65},
  year={2007}
}

@book{rudin1987real,
  title={Real and complex analysis},
  author={Rudin, Walter},
  year={1987},
  publisher={McGraw-Hill, Inc.}
}

\end{document}